%% file: main.tex
\documentclass[runningheads]{llncs}
\usepackage{graphicx}
\usepackage{hyperref} 
\usepackage{verbatim} 
\usepackage{xcolor} 
\usepackage{mathtools}
\usepackage{amssymb}
\usepackage{amsmath}
\usepackage{stmaryrd}
\usepackage{xspace}
\usepackage{listings}
\usepackage{color}

\newif\ifsubmit
\submittrue

\ifsubmit
\newcommand{\EZ}[1]{{#1}} 

\newcommand{\PB}[1]{{#1}}
\newcommand{\DA}[1]{{#1}} 
\newcommand{\DAComm}[1]{} 
\newcommand{\EZComm}[1]{} 
\newcommand{\FDComm}[1]{}
\newcommand{\PBComm}[1]{}
\else
\newcommand{\EZ}[1]{\textcolor{blue}{#1}} 
 
\newcommand{\PB}[1]{\textcolor{purple}{#1}} 
\newcommand{\DA}[1]{\textcolor{red}{#1}} 
\newcommand{\DAComm}[1]{{\scriptsize\textcolor{red}{[\bf{Davide: }#1}]}}
\newcommand{\EZComm}[1]{{\scriptsize\textcolor{blue}{[\bf{Elena: }#1}]}}
\newcommand{\FDComm}[1]{{\scriptsize\textcolor{orange}{[\bf{Francesco: }#1}]}}
\newcommand{\PBComm}[1]{{\scriptsize\textcolor{purple}{[\bf{Pietro: }#1}]}}
\fi

\input{macros}

\begin{document}

\title{Enhanced Regular Corecursion\\
 for Data Streams}
%
%
\author{Davide Ancona \and Pietro Barbieri \and Elena Zucca}
\authorrunning{D. Ancona et al.}
%
\institute{DIBRIS, University of Genova
} 
\maketitle              
\begin{abstract}
  We propose a simple calculus for processing \emph{data streams} (infinite flows of data series), represented by finite sets of equations built on stream operators. Furthermore, functions defining streams are \emph{regularly corecursive}, that is, cyclic calls are detected, avoiding non-termination as happens with ordinary recursion in the call-by-value evaluation strategy. As we illustrate by several examples, the combination of such two mechanisms provides a good compromise between expressive power and decidability.
Notably, we provide an algorithm to check that the stream returned by a function call is represented by a \emph{well-formed set of equations} which actually admits a unique solution, hence access to an arbitrary element of the returned stream will never diverge.
\keywords{Operational semantics, stream programming, regular terms.}
\end{abstract}
\input{intro}

\input{streamSem}

\input{examples}
\input{wf_iff}
\input{conclu}


\bibliographystyle{plain}
\bibliography{bib}
\newpage
\appendix
\input{appendix}

\end{document}

%% file: macros.tex
\newcommand{\refToFigure}[1]{Fig.~\ref{fig:#1}}
\newcommand{\refToSection}[1]{Sect.~\ref{sect:#1}}
\newcommand{\refToTheorem}[1]{Theorem~\ref{theo:#1}}
\newcommand{\refToLemma}[1]{Lemma~\ref{lemma:#1}}

\newcommand{\refToRule}[1]{{\small \textsc{(#1)}}}
\newcommand{\Space}{\hskip 0.8em}
\newcommand{\BigSpace}{\hskip 1.5em}

\newcommand{\dom}[1]{\mathit{dom}(#1)}
\newcommand{\Tuple}[1]    {({#1})}
\newcommand{\Pair}[2]     {\Tuple{{#1},{#2}}}
\newcommand{\fv}[1]{\textit{fv}(#1)}
\newcommand{\vars}[1]{\textit{vars}(#1)}
 
\newcommand{\NamedRule}[4]{{\tiny\textsc{({#1})}}\displaystyle\frac{#2}{#3}\ \begin{array}{l} #4 \end{array}}
\newcommand{\NamedRuleSimple}[3]{{\tiny\textsc{({#1})}}\displaystyle\frac{#2}{#3}}
\newenvironment{grammatica}{$\begin{array}[t]{llll}}{\end{array}$}
\newcommand{\produzione}[3]{#1&{:}{:}=&#2 & \mbox{{\small{#3}}}}
\newcommand{\terminale}[1]{\texttt{#1}}
\newcommand{\nonterminale}[1]{\textit{#1}}

\newcommand{\fd}{\nonterminale{fd}}
\newcommand{\f}{\nonterminale{f}}
\newcommand{\undef}{\nonterminale{undef}}
\newcommand{\g}{\nonterminale{g}}
\newcommand{\fdBar}{\overline{\fd}}
\newcommand{\x}{\nonterminale{x}} 
\newcommand{\y}{\nonterminale{y}} 
\newcommand{\z}{\nonterminale{z}} 
\newcommand{\xBar}{\overline{\x}}
\newcommand{\E}{\nonterminale{e}}
\newcommand{\EBar}{\overline{\E}}
\newcommand{\s}{\nonterminale{se}}
\newcommand{\n}{\nonterminale{ne}}
\newcommand{\be}{\nonterminale{be}}
\newcommand{\true}{\texttt{true}}
\newcommand{\false}{\texttt{false}}
\newcommand{\cons}[2]{{#1}\mathbin{\terminale{:}}{#2}}
\newcommand{\tail}[1]{{#1}{\char`\^}}
\newcommand{\IfThenElse}[3] {{\terminale{if}\mathrel{#1}\terminale{then}}\mathrel{#2}{\terminale{else}\mathrel{#3}}}
\newcommand{\PW}[3]{#1{#2}#3}
\newcommand{\nop}{\nonterminale{op}} 
\newcommand{\pw}[1]{[{#1}]}
\newcommand{\pwop}{\pw{\nop}}

\newcommand{\Caps}[2]{\Pair{#1}{#2}}

\newcommand{\call}{\mathit{c}}
\newcommand{\val}{\nonterminale{v}}
\newcommand{\vBar}{\overline{\val}}
\newcommand{\callEnv}{\tau}
\newcommand{\emptyMap}{\ensuremath{\emptyset}}
\newcommand{\mapEnv}{\rho}
\newcommand{\mapEnvPrime}{{\mapEnv'}}
\newcommand{\mapEnvU}{{\widehat{\mapEnv}}}
\newcommand{\sv}{\nonterminale{s}}
\newcommand{\nv}{\nonterminale{n}}
\newcommand{\bv}{\nonterminale{b}}
\newcommand{\opsem}[5]{{#1,#2,#3}\!\Downarrow\!\Pair{#4}{#5}}

\newcommand{\bigmapUnion}[2]{\bigsqcup_{#1} #2}
\newcommand{\Subst}[3]   {#1 [#2/#3]}
\newcommand{\Update}[3]   {#1 \{#2\mapsto#3\}}
\newcommand{\bisim}[3]{\EZ{ #1 \approx_{#3} #2}} 
\newcommand{\uptobisim}[2]{#1_{\approx #2}}
\newcommand{\At}[4]{\atfun_#1(#2,#3)=#4}
\newcommand{\atfun}{\mathit{at}}
\newcommand{\wf}{\mathit{wd}}
\newcommand{\indx}{i}
\newcommand{\jndx}{j}

\newcommand{\stream}{\sigma} 
\newcommand{\eval}[2]{\sem{#1}#2}
\newcommand{\subst}{\theta}
\newcommand{\sol}{\textsf{sol}}

\newcommand{\AppMap}[2]{ #1[#2] }
\newcommand{\sem}[1]{{\llbracket#1\rrbracket}}

\newcommand{\WF}[2]{\mathsf{wd}(#1,\mapEnv,#2)}
\newcommand{\WFUpdated}[3]{\mathsf{wd}(#1,#2,#3)}
\newcommand{\map}{\mathit{m}}
\newcommand{\incrMap}{\map^{+1}}
\newcommand{\decrMap}{\map^{-1}}

\newcommand{\AT}[2]{\atfun_\mapEnv(#1,#2)}
\newcommand{\premise}{\vdash}
\newcommand{\premisestar}{\vdash^\star}
\newcommand{\premisestarfirst}[1]{\vdash^\star_{#1}}
\newcommand{\varset}{\mathcal{X}}


%% file: intro.tex
\section{Introduction}\label{sect:intro}

Applications often deal with data structures which are conceptually infinite, among those \emph{data streams} ({infinite flows of data series}) are a mainstream example: as we venture deeper into the Internet of Things (IoT) era, stream processing is becoming increasingly important. Indeed, all
  main IoT platforms provide embedded and integrated engines for real time analysis of potentially infinite flowing data series; such a process occurs before the data is
  stored for efficiency and, as often happens in Computer Science, there is a trade-off between the expressive power of the language, 
  the efficiency of its implementation and the decidability of properties important  to guarantee reliability and tractability.

Another related important problem is data stream generation, which is essential to test complex distributed IoT systems;
  the deterministic simulation of sensor data streams through a suitable language offers a practical solution to IoT testing and favors
  early detection of some kinds of bugs that can be fixed more easily before the deployment of the whole system.
  

A well-established solution {to data stream generation and processing} is \emph{lazy evaluation}, as supported, e.g., in Haskell,
{and most stream libraries offered by mainstream languages, as \lstinline{java.util.stream}}. In this approach, conceptually infinite
data streams are the result of a function or method call, which is evaluated {according to the call-by-need strategy}.
For instance, {in Haskell} we can define \lstinline{one_two = 1:2:one_two}, or even represent the list of natural numbers {as \lstinline{from 0}}, where \lstinline{from n = n:from(n+1)}.
{However, such a great expressive power comes at a cost; let us consider, for instance, the definition \lstinline{bad_stream = 0:tail bad_stream}.
  The Haskell compiler does not complain about this definition, and no problem arises at runtime as long as the manipulation of \lstinline{bad_stream} requires only its first element to be accessed; anyway, any operation which needs to inspect \lstinline{bad_stream} at a deeper level is deemed to diverge.
  Unfortunately, it is not decidable to check, even at runtime, whether the stream returned by a Haskell function is \emph{\DA{well-defined}},
  that is, all of its elements can be computed\footnote{\DA{This is what is also known as a productive corecursive definition \cite{Coquand93}.}}; {indeed,} the full expressive power of Haskell can be used
  to define streams by means of recursive functions. For similar reasons, it is not decidable to check at runtime whether the streams returned by two Haskell functions are equal.}

More recently, a complementary approach has been considered in different programming paradigms --- functional \cite{Jeannin17}, logic \cite{SimonBMG07,Ancona13,DagninoAZ20}, and object-oriented \cite{AnconaBDZ20} --- based on the following two ideas:
\begin{itemize}
\item Infinite streams can be finitely represented by \emph{finite sets of equations} involving only the stream constructor, e.g., $\x = 1 : 2 : \x$.
  \DA{Such a representation corresponds to what has been called by Courcelle in its seminal paper \cite{Courcelle83} a \emph{regular}, a.k.a. \emph{rational},
    tree}, that is, a tree with possibly infinite depth but a finite set of subtrees.
\item Functions are \emph{regularly corecursive}, that is, execution keeps track of pending function calls, so that, when the same call is considered the second time, this is detected, avoiding non-termination {as happens with ordinary recursion in the call-by-value evaluation strategy}. 
\end{itemize} 
In this way, {the Haskell stream  \lstinline{one_two} can be equivalently {obtained} by} the call\footnote{Differently from Haskell, for simplicity in our calculus functions {are uncurried, hence they take as arguments possibly empty tuples, delimited by parentheses.}}  \lstinline{one_two()}, {with the function \lstinline{one_two}} defined by \lstinline{one_two() = 1:2:one_two()}.
Indeed, with regular corecursion the result of this call is the value corresponding to the unique solution of the equation $\x = 1 : 2 : \x$.
  On the other hand, since the expressive power is limited to regular streams, it is not possible to define a corecursive function whose call
  returns the stream of natural numbers, as happens for the \lstinline{from 0} Haskell example. However, \DA{there exist procedures for
  checking well-defined streams and their equality, even with tractable algorithms}.
 
In this paper, we propose a simple calculus of {numeric} streams which supports {regular corecursion and goes beyond regular streams
  by extending equations with other typical stream operators besides the stream constructor: tail and pointwise operators can be contained in stream equations and are therefore not evaluated.

In this way, we {are able to} achieve a good compromise between expressive power and {decidability}. Notably:
\begin{itemize}
\item {the extended {shape} of equations allows {the definition of functions} which return non-regular streams; for instance, it is possible to {obtain} the stream of natural numbers as {\lstinline!from(0)!}, {by defining} \lstinline{from(n)=n:(from(n)[+]repeat(1))}, with \lstinline{[+]} the pointwise addition on numeric streams and \lstinline{repeat} the function defined by
\lstinline{repeat(n)=n:repeat(n)};}

\item {there exists a decidable procedure to dynamically check whether the stream returned by a corecursive function is \DA{well-defined}}; 

\item {however, it is not possible to express} \emph{all} streams computable with the lazy evaluation approach, but only those which have a specific structure (that is, can be expressed as the unique solution of a set of equations built with the above mentioned operators).
\end{itemize}

In \refToSection{streamSem} we formally define the calculus, in \refToSection{examples} we show many interesting examples, and in \refToSection{wf} we provide an operational characterization of \DA{well-defined streams}, proved to be a sufficient and necessary condition for an access to an arbitrary index to never diverge.
In \refToSection{conclu} we discuss related and further work. The Appendix contains more examples of derivations.

%% file: streamSem.tex
\section{Stream calculus}\label{sect:streamSem}
 \refToFigure{stream-syntax} shows the syntax of the calculus. 

\begin{figure}
\begin{small}
\begin{grammatica}
\produzione{\fdBar}{\fd_1\ldots\fd_n}{program}\\
\produzione{\fd}{\f(\xBar) = \s}{function declaration}\\
\produzione{\E}{\s \mid \n \mid \be}{expression}\\
\produzione{\s}{\x \mid \IfThenElse{\be}{\s_1}{\s_2} \mid \cons{\n}{\s} \mid \tail{\s} \mid \PW{\s_1}{\pwop}{\s_2} \mid \f(\EBar)}{stream expression}\\
\produzione{\n}{\x \mid \s(\n) \mid \n_1\mathop{\nop}\n_2 \mid 0 \mid 1 \mid 2 \mid ...}{numeric expression}\\
\produzione{\be}{\x \mid \true \mid \false \mid ...}{boolean expression}\\
\produzione{\nop}{+\ \mid\ -\ \mid\ *\ \mid\ /}{numeric operation}
\end{grammatica}
\end{small}
\caption{Stream calculus:  syntax}
\label{fig:stream-syntax}
\end{figure}

A program is a sequence of (mutually recursive) function declarations, for simplicity assumed to only return streams. Stream expressions are variables, conditional expressions, expressions built by stream operators, and function calls. We consider the following stream operators: constructor (prepending a numeric element), tail, and pointwise {arithmetic} operations. Numeric expressions  include the access to the $i$-th\footnote{{For simplicity, here indexing and numeric expressions coincide, even though indexes are expected to be natural numbers, while values in streams can range over a larger numeric domain}.}} element of a stream. {We use $\fdBar$ to denote a sequence $\fd_1, \dots, \fd_n$ of function declarations, and analogously for other sequences.}

The operational semantics, given in \refToFigure{stream-sem}, is based on two key ideas:
\begin{enumerate}
\item (some) infinite streams are represented in a finite way
\item evaluation keeps trace of already considered function calls 
\end{enumerate}

\begin{figure}
\begin{small}
\begin{grammatica}
\produzione{\call}{\f(\vBar)}{{(evaluated)} call}\\
\produzione{\val}{\sv \mid \nv \mid \bv}{value}\\
\produzione{\sv}{\x \mid \cons{\nv}{\sv} \mid \tail{\sv} \mid \PW{\sv_1}{\pwop}{\sv_2} }{{(open)} stream value}\\
\produzione{{\indx, \nv}}{0 \mid 1 \mid 2 \mid ...}{{index, numeric value}}\\
\produzione{\bv}{\true \mid \false}{boolean value}\\
\produzione{\callEnv}{\call_1\mapsto\x_1\ \ldots\ \call_n\mapsto\x_n \Space (n\geq0)}{call trace}\\
\produzione{\mapEnv}{{\x_1\mapsto\sv_1 \ldots \x_n\mapsto\sv_n} \Space (n\geq0)}{environment}\\
\end{grammatica}
\\[2ex]

\hrule 

$\begin{array}{l}
  \\
\NamedRule{val}{}{\opsem{\val}{\mapEnv}{\callEnv}{\val}{\mapEnv}}{}
\Space
\NamedRule{if-t}{
\opsem{\be}{\mapEnv}{\callEnv}{\true}{\mapEnv} \Space
\opsem{\s_1}{{\mapEnv}}{\callEnv}{\sv}{\mapEnv'}
}{\opsem{\IfThenElse{{\be}}{\s_1}{\s_2}}{\mapEnv}{\callEnv}{\sv}{{\mapEnv'}}
}{}
\Space
\NamedRule{if-f}{
\opsem{\be}{\mapEnv}{\callEnv}{\false}{{\mapEnv}} \Space
\opsem{\s_2}{{\mapEnv}}{\callEnv}{\sv}{{\mapEnv'}}
}{\opsem{\IfThenElse{{\be}}{\s_1}{\s_2}}{\mapEnv}{\callEnv}{\sv}{{\mapEnv'}}
}{}
\\[6ex]
\NamedRule{cons}{
\opsem{\n}{\mapEnv}{\callEnv}{\nv}{\mapEnv}\Space
\opsem{\s}{\mapEnv}{\callEnv}{\sv}{\mapEnv'}}{
\opsem{\cons{\n}{\s}}{\mapEnv}{\callEnv}{\cons{\nv}{\sv}}{\mapEnv'}}
{} \BigSpace
\NamedRule{tail}{
\opsem{\s}{\mapEnv}{\callEnv}{\sv}{\mapEnv'}}{
\opsem{\tail{\s}}{\mapEnv}{\callEnv}{\tail{\sv}}{\mapEnv'}}
{}
\BigSpace
\NamedRule{pw}{
\opsem{\s_1}{\mapEnv}{\callEnv}{\sv_1}{\mapEnv_1}\Space
\opsem{\s_2}{\mapEnv}{\callEnv}{\sv_2}{\mapEnv_2}}{
\opsem{\PW{\s_1}{\pwop}{\s_2}}{\mapEnv}{\callEnv}{\PW{\sv_1}{\pwop}{\sv_2}}{\mapEnv_1\sqcup\mapEnv_2}}
{}
\\[6ex]
{\NamedRule{args}{  \begin{array}{l}
    \opsem{\E_i}{\mapEnv}{\callEnv}{\val_i}{{\mapEnv_i}}\Space \forall i \in 1..n
    \BigSpace
    \opsem{\f(\vBar)}{{\mapEnvU}}{\callEnv}{\sv}{\mapEnvPrime}
  \end{array}
}{ \opsem{\f(\EBar)}{\mapEnv}{\callEnv}{\sv}{\mapEnvPrime}}
{ \EBar=\E_1,\ldots,\E_n\ \mbox{not of shape}\ \vBar\\
\vBar=\val_1,\ldots,\val_n\\
\mapEnvU = \bigmapUnion{i \in 1..n}{\mapEnv_i}}}

\\[6ex]
{\NamedRule{invk}{  \begin{array}{l}
    \opsem{\Subst{\s}{\vBar}{\xBar}}{\mapEnv}{\Update{\callEnv}{\f(\vBar)}{\x}}{\sv}{\mapEnvPrime}
  \end{array}
}{ \opsem{\f(\vBar)}{\mapEnv}{\callEnv}{{\x}}{\Update{\mapEnvPrime}{\x}{\sv}}}
{ {\f(\vBar)\not\in\dom{\uptobisim{\callEnv}{\mapEnv}}}\\
 \x\ \mbox{fresh}\\ 
\mathit{fbody}(\f)=\Pair{\xBar}{\s}\\
\wf(\mapEnvPrime,\x,\sv)}}
\BigSpace
{\NamedRule{corec}{
}{ \opsem{\f(\vBar)}{\mapEnv}{\callEnv}{\x}{\mapEnv}}  
{
\uptobisim{\callEnv}{\mapEnv}({f(\vBar)})=\x\\
}}
\\[8ex]
\NamedRule{at}{
  \opsem{\s}{\mapEnv}{\callEnv}{\sv}{{\mapEnv'}}\Space  
  \opsem{\n}{\mapEnv}{\callEnv}{\indx}{\mapEnv}\Space
}{ \opsem{\s(\n)}{\mapEnv}{\callEnv}{\nv}{ {\mapEnv} } }
{  \At{\mapEnvPrime}{\sv}{\indx}{\nv}}
\\[6ex]
\hline
\\
\NamedRule{at-var}
{\At{\mapEnv}{\mapEnv(\x)}{\indx}{\nv'}}
{\At{\mapEnv}{\x}{\indx}{\nv' }}
{  }
\BigSpace
\NamedRule{at-cons-0}
{}
{\At{\mapEnv}{\cons{\nv}{\sv}}{0}{\nv }}
{  }
\BigSpace
\NamedRule{at-cons-n}
{\At{\mapEnv}{\sv}{\indx-1}{\nv'}}
{\At{\mapEnv}{\cons{\nv}{\sv}}{\indx}{\nv'}}
{\indx>0}
\\[6ex]
\NamedRule{at-tail}
{\At{\mapEnv}{\sv}{\indx+1}{\nv}}
{\At{\mapEnv}{\tail{\sv}}{\indx}{\nv}}
{}
\BigSpace
\NamedRule{{at-pw}}
{\At{\mapEnv}{\sv_1}{\indx}{\nv_1}\Space{\At{\mapEnv}{\sv_2}{\indx}{\nv_2}}}
{\At{\mapEnv}{\PW{\sv_1}{\pwop}{\sv_2}}{\indx}{\nv_1\mathbin{\nop}\nv_2}}
{}

\end{array}$
\end{small}
\caption{Stream calculus: operational semantics}\label{fig:stream-sem}
\end{figure}

To obtain (1), our approach is inspired by \emph{capsules} \cite{JeanninK12}, which are essentially expressions supporting cyclic references. That is, the \emph{result} of the evaluation of a stream expression is a pair $\Pair{\sv}{\mapEnv}$, where $\sv$ is an \emph{(open) stream value}, built on top of stream variables, numeric values, the stream constructor, the tail destructor and the pointwise arithmetic operators, and $\mapEnv$ is an \emph{environment} mapping a finite set of variables into stream values.
In this way, cyclic streams can be obtained: for instance, $\Pair{\x}{\x\mapsto\cons{\nv}{\x}}$ denotes the stream constantly equal to 
 $\nv$.
 
We denote by $\vars{\mapEnv}$ the set of variables occurring in $\mapEnv$, by $\fv{\mapEnv}$ the set of its free variables, that is, $\vars{\mapEnv}\setminus\dom{\mapEnv}$, and say that $\mapEnv$ is \emph{closed} if $\fv{\mapEnv}=\emptyset$, \emph{open} otherwise, and analogously for a result $\Pair{\val}{\mapEnv}$.

%

To obtain {point (2) above}, evaluation has an additional parameter which is a \emph{call trace}, a map from function calls where arguments are values  (dubbed \emph{calls} for short in the following) into variables. 

Altogether, the semantic judgment has shape $\opsem{\E}{\mapEnv}{\callEnv}{\val}{\mapEnvPrime}$, {where 
$\E$ is the expression to be evaluated, $\mapEnv$ the current environment defining possibly cyclic stream values that can occur in $\E$,
  $\callEnv$  the call trace, and $\Pair{\val}{\mapEnvPrime}$ the result.}
  The semantic judgments should be indexed by an underlying (fixed) program, 
 omitted for sake of simplicity. Rules use the following auxiliary definitions:
\begin{itemize}
\item $\mapEnv\sqcup\mapEnv'$ is the union of two environments, which is well-defined if they have disjoint domains; $\Update{\mapEnv}{\x}{\sv}$ is the environment which gives $\sv$ on $\x$, coincides with $\mapEnv$ elsewhere; we use analogous notations for call traces.
\item $\Subst{\s}{\vBar}{\xBar}$ is obtained by {parallel substitution of} variables $\xBar$ {with} values $\vBar$.
\item $\mathit{fbody}(\f)$ returns the pair of the parameters and the body of the declaration of $\f$, if any, {in the assumed program}.
\end{itemize}
Moreover, the rules are parametric in the following other judgments, for which different definitions will be discussed {in \refToSection{wf}:}
{\begin{itemize}
\item $\wf(\mapEnv,\x,\sv)$, that is, by adding the association $\x\mapsto\sv$ to the (\DA{well-defined}) environment $\mapEnv$,  we still get a \DA{well-defined} environment.
\item $\bisim{\val}{\val'}{\mapEnv}$, that is, the two values are equivalent in the given environment\footnote{This equivalence is assumed to be the identity on numeric and boolean values.}. Then, $\uptobisim{\callEnv}{\mapEnv}$ is \DA{the extension of $\callEnv$ \emph{up to {equivalence}} in $\mapEnv$:} $\uptobisim{\callEnv}{\mapEnv}({f(\val_1,\ldots,\val_n)})=\x$ iff there exist $\val'_1,\ldots,\val'_n$ such that
$\callEnv(f(\val'_1,\ldots,\val'_n))=\x$ and $\bisim{\val_i}{\val'_i}{\mapEnv}$ for all $i\in1..n$.
\end{itemize}}

Intuitively, a closed result $\Pair{\sv}{\mapEnv}$ is \DA{well-defined} if it denotes a unique stream (infinite sequence of numeric values), and a closed environment $\mapEnv$ is \DA{well-defined} if, for each $\x\in\dom{\mapEnv}$, $\Pair{\x}{\mapEnv}$ is \DA{well-defined}. In other words, the corresponding set of equations admits a unique solution. For instance, the {environment
  ${\{\x\mapsto\x\}}$} is not \DA{well-defined}, since it is undetermined (any stream satisfies the equation {$\x=\x$}); the {environment $\{\x\mapsto\x[+]\y,\y\mapsto\cons{1}{\y}\}$} is not \DA{well-defined} as well, since it is undefined ({the two equations $\x=\x\mapsto\x[+]\y,\y=\cons{1}{\y}$ admit no solutions for $x$}). Finally, two stream values $\sv$ and $\sv'$ such that the results $\Pair{\sv}{\mapEnv}$ and $\Pair{\sv'}{\mapEnv}$ are closed and \DA{well-defined} are equivalent if they denote the same stream.

These notions can be generalized to open results and environments, assuming that free variables denote unique streams, as will be formalized in \refToSection{wf}. 

Rules for values and conditional are straightforward. In rules \refToRule{{cons}}, \refToRule{tail} and \refToRule{pw}, arguments are evaluated,
\DA{while} the stream operator is applied without any further evaluation; \DA{the fact that the tail and pointwise operators are treated
as the stream constructor $\cons{\_}{\_}$ is crucial to get results which denote non-regular streams as shown in \refToSection{examples}. However,
when non-constructors are allowed to occur in values, ensuring well-defined results become more challenging, because
the usual simple syntactic constraints that can be safely used for constructors \cite{Coquand93} no longer work (see more details in \refToSection{wf} and
\ref{sect:conclu}).}

The rules for function call are based on a mechanism of cycle detection, similar to that in \cite{AnconaBDZ20}. They are given in a modular way. That is, evaluation of arguments is handled by a separate rule \refToRule{args}, whereas the following two rules handle (evaluated) calls.

Rule \refToRule{invk} is applied when a call is considered for the first time, as expressed by the first side condition. The body is retrieved by using the auxiliary function \textit{fbody}, and {evaluated in  a call trace where the call has been mapped into a fresh variable. Then, it is checked that adding  the association from such variable to the result of the evaluation of the body keeps the environment \DA{well-defined}.
If the check succeeds, then the final result consists of the variable associated with the call and the updated environment. For simplicity, here execution is stuck if the check fails; an implementation should raise a runtime error instead.

Rule \refToRule{corec} is applied when a  call is considered for the second time, as expressed by the first side condition (note that cycle detection takes place up to equivalence in the environment). The variable $\x$ is returned as result. {However, there is no associated value in the environment yet; in other words, the {result $\Pair{\x}{\mapEnv}$ is open} at this point. } This means that $\x$ is undefined until the environment is updated with the corresponding value in rule \refToRule{invk}. However, $\x$ can be safely used as long as the evaluation does not require $\x$ to be inspected; for instance, $\x$ can be safely passed as an argument to a function call.

For instance, if we consider the program \lstinline!f()=g()  g()=1:f()!, then the judgment
  $\opsem{\mbox{\lstinline{f()}}}{\emptyMap}{\emptyMap}{\x}{\mapEnv}$, with $\mapEnv=\{\x\mapsto\y,\y\mapsto \cons{1}{\x}\}$,  is derivable;
  however, while the final result $\Pair{\x}{\mapEnv}$ is closed, the derivation contains also judgments with open results, as happens for
$\opsem{\mbox{\lstinline{f()}}}{\emptyMap}{\{\mbox{\lstinline{f()}}\mapsto\x,\mbox{\lstinline{g()}}\mapsto\y\}}{\x}{\emptyMap}$ and $\opsem{\mbox{\lstinline{g()}}}{\emptyMap}{\{\mbox{\lstinline{f()}}\mapsto\x\}}{\y}{\{\y\mapsto\cons{1}{\x}\}}$. For the full derivation, see \refToFigure{derivation1} in the appendix.

As another example, if we consider the program \lstinline!f()=g(2:f())  g(s)=1:s!, then the derivation of the judgment
$\opsem{\mbox{\lstinline{f()}}}{\emptyMap}{\emptyMap}{\x}{\mapEnv}$ with $\mapEnv=\{\x\mapsto\y,\y\mapsto \cons{1}{\cons{2}{\x}}\}$
  is built on top of the derivation of $\opsem{\mbox{\lstinline{g(2:}\x\lstinline{)}}}{\emptyMap}{\{\mbox{\lstinline{f()}}\mapsto\x\}}{\y}{\{\y\mapsto\cons{1}{\cons{2}{\x}}\}}$, corresponding to the evaluation of \lstinline{g(2:}\x\lstinline{)} where $\x$ is an operand of the stream constructor
whose result is passed as argument to the call to \lstinline{g}, despite $\x$ is not defined yet. For the full derivation, see \refToFigure{derivation2} in the appendix.

Finally, rule \refToRule{at} computes the $\indx$-th element of a stream expression. After evaluation of the arguments, the numeric result is obtained by the auxiliary judgment $\At{\mapEnv}{\sv}{\indx}{\nv}$, inductively defined in the bottom part of the figure.  If the stream value is a variable (\refToRule{at-var}), then the evaluation is propagated to the associated stream value in the environment, if any. If, instead, the variable is free in the environment, then execution is stuck; again, an implementation should raise a runtime error instead. 
\refToFigure{stuck_derivation} in the appendix shows an example of stuck derivation. 
If the stream value is built by the constructor, then the result is the first element of the stream if the index is $0$ (\refToRule{at-cons-0}); otherwise, the {evaluation} is recursively {propagated to} its tail
{with} the predecessor {index} (\refToRule{at-cons-n}). Conversely, if the stream is built by the tail operator (\refToRule{at-tail}), then
the {evaluation} is recursively {propagated to} the stream argument {with} the successor {index}. Finally, if the stream is built by a pointwise operation (\refToRule{at-pw}), then {the evaluation is recursively propagated to the operands with the same index and then the corresponding arithmetic operation
is computed on the results}.

%% file: examples.tex
\section{Examples}\label{sect:examples}
First we show some simple examples, to explain how regular corecursion works. Then we provide some more significant examples.

Consider the following function declarations:
\begin{lstlisting}
repeat(n) = n:repeat(n) 
one_two() = 1:two_one()
two_one() = 2:one_two()
\end{lstlisting}
With the standard semantics of recursion, the calls, e.g., \lstinline{repeat(0)} and \lstinline{one_two()} lead to non-termination. Thanks to regular corecursion, instead, these calls terminate, producing as result $\mathtt{\Caps{\x}{\{\x\mapsto\cons{0}{\x}\}}}$, and $\mathtt{\Caps{\x}{\{\x\mapsto\cons{1}{\y},\y\mapsto\cons{2}{\x}\}}}$, respectively. Indeed, when initially invoked, the call  \lstinline{repeat(0)} is added in the call trace with an associated fresh variable, say $\x$. In this way, when evaluating the body of the function, the recursive call is detected as cyclic, the variable $\x$ is returned as its result, and, finally, the stream value $\mathtt{\cons{0}{\x}}$ is associated in the environment with the result $\x$ of the initial call. The evaluation of \lstinline{one_two()} is analogous, except that another fresh variable $\y$ is generated for the intermediate call \lstinline{two_one()}. The formal derivations are given below.

\begin{small}
$$\begin{array}{l}
\NamedRuleSimple{invk}{
  \NamedRuleSimple{cons}{
  \NamedRuleSimple{value}{}{}\BigSpace
    \NamedRuleSimple{corec}{
}{\opsem{\mathtt{repeat(0)}}{\emptyset}{\{\mathtt{repeat(0)}\mapsto\x\}}{\x}{\emptyset}}
}{\opsem{\mathtt{\cons{0}{repeat(0)}}}{\emptyset}{\{\mathtt{repeat(0)}\mapsto\x\}}{\cons{0}{\x}}{\emptyset}}}
{\opsem{\mathtt{repeat(0)}}{\emptyset}{\emptyset}{\x}{\{\x\mapsto\cons{0}{\x}\}}}
\end{array}$$

$$\begin{array}{l}
\NamedRuleSimple{invk}{
  \NamedRuleSimple{cons}{
  \NamedRuleSimple{value}{}{}\BigSpace
\NamedRuleSimple{invk}{
  \NamedRuleSimple{cons}{
  \NamedRuleSimple{value}{}{}\BigSpace
\NamedRuleSimple{corec}{
}{\opsem{\mathtt{one\_two()}}{\emptyset}{\{\mathtt{one\_two()}\mapsto\x,\ \mathtt{two\_one()}\mapsto\y\}}{\x}{\emptyset}}
}{\opsem{\mathtt{2:one\_two()}}{\emptyset}{\{\mathtt{one\_two()}\mapsto\x,\ \mathtt{two\_one()}\mapsto\y\}}{2:\x}{\emptyset}}
}{\opsem{\mathtt{two\_one()}}{\emptyset}{\{\mathtt{one\_two()}\mapsto\x\}}{\y}{\{\y\mapsto2:\x\}}}
}{\opsem{\mathtt{1:two\_one()}}{\emptyset}{\{\mathtt{one\_two()}\mapsto\x\}}{1:\y}{\{\y\mapsto2:\x\}}}}
                {\opsem{\mathtt{one\_two()}}{\emptyset}{\emptyset}{\x}{\{\x \mapsto 1:\y,\ \y\mapsto2:\x\}}}
\end{array}$$
\end{small}

For space reasons, we did not report the application of rule \refToRule{value}. In both derivations, note that rule \refToRule{corec} is applied, without evaluating the body once more, when the cyclic call is detected. 

The following examples show function definitions {whose calls return} non-regular streams, notably, the natural numbers, {the natural numbers raised to the power of a number, the factorials, the powers of a  number,} the Fibonacci numbers, and the stream obtained by pointwise increment by one.
\begin{lstlisting}  
nat() = 0:(nat()[+]repeat(1))
nat_to_pow(n) =                  //nat_to_pow(n)(i)=i^n 
  if n <= 0 then repeat(1) else nat_to_pow(n-1)[*]nat()
fact() = 1:((nat()[+]repeat(1))[*]fact())
pow(n) = 1:(repeat(n)[*]pow(n)) //pow(n)(i)=n^i
fib() = 0:1:(fib()[+]fib()^)
incr(s) = s[+]repeat(1)
\end{lstlisting}

The definition of \lstinline{nat} uses regular corecursion, since the recursive call \lstinline{nat()} is cyclic. Hence the call \lstinline{nat()} returns $\Caps{\x}{\{\x\mapsto\cons{0}{(\x[+]\y)}, \y\mapsto\cons{1}{\y}\}}$.
{The definition of \lstinline{nat_to_pow} is} a standard inductive one where the argument strictly decreases in the recursive call. Hence, the call, e.g., \lstinline{nat_to_pow(2)}, returns\\
\centerline{
$\Caps{\x_2}{\{\x_2\mapsto\x_1[*]\x,\x_1\mapsto\x_0[*]\x, \x_0\mapsto\y, \y\mapsto\cons{1}{\y}, \x\mapsto\cons{0}{(\x[+]\y')}, \y'\mapsto\cons{1}{\y'}\}}.$}
  The definitions of \lstinline{fact}, \lstinline{pow}, and \lstinline{fib} are regularly corecursive. For instance, the call \lstinline{fact()} returns
  $\Caps{\z}{\z\mapsto(\x[+]\y)[*]z, \x\mapsto\cons{0}{(\x[+]\y'}), \y\mapsto 1:\y, \y'\mapsto 1:\y'}$. 
 The definition of \lstinline{incr} is non-recursive, hence always converges{, and the call \lstinline{incr(}$\sv$\lstinline{)} returns $\Caps{\x}{\{\x\mapsto \sv[+]\y, \y\mapsto\cons{1}{\y}\}}$}.
The following alternative definition
\begin{lstlisting} 
incr_reg(s) = (s(0)+1):incr_reg(s^) 
\end{lstlisting}
relies, instead, on regular corecursion. Note the difference: the latter version ensures termination only for
  regular streams, {as in  \lstinline{incr_reg(one_two())}, since, eventually, in the recursive call, the expression \lstinline{s^} turns out to denote the initial stream}; however, the computation does not
  terminate for non-regular streams, as in \lstinline{incr_reg(nat())}, which, however, converges with \lstinline{incr}. 
  
The following function computes the stream of partial sums of the first $i+1$ elements of a stream $s$, that is,
  \lstinline!sum($s$)($\indx$)$=\sum_{k=0}^{\indx}s(k)$!:
\begin{lstlisting}
sum(s) = s(0):(s^[+]sum(s))
\end{lstlisting}
Such a function is useful for computing streams whose elements approximate a series with increasing precision;
for instance, the following function returns the stream of partial sums of the first $i+1$ elements of the Taylor series of the exponential function:
\begin{lstlisting}
sum_expn(n) = sum(pow(n)[/]fact())
\end{lstlisting}
Function \lstinline{sum_expn} calls \lstinline{sum} with the argument \lstinline{pow(n)[/]fact()} corresponding to the stream of all terms of the
Taylor series of the exponential function;  hence, by accessing the $\indx$-th element of the stream, we have the following approximation of the series:}
\begin{quote}
\texttt{sum\_expn(}$\nv$\texttt{)(}$\indx$\texttt{)}$=\sum_{k=0}^{\indx} \frac{{\nv}^k}{k!} = 1+\nv+\frac{{\nv}^2}{2!}+\frac{{\nv}^3}{3!}+\frac{{\nv}^4}{4!}+\cdots+\frac{\nv^{\indx}}{\indx!}$
\end{quote}
Lastly, we present a couple of examples showing how it is possible to define primitive operations provided
  by IoT platforms for real time analysis of data streams; we start with \lstinline{aggr(n,s)}, which allows
  aggregation (by addition) of contiguous data in the stream \lstinline{s} w.r.t. a frame of length \lstinline{n}:
\begin{lstlisting}
aggr(n,s) = if n<=0 then repeat(0) else s[+]aggr(n-1,s^) 
\end{lstlisting}

For instance, \lstinline!aggr(3,$\sv$)! returns the stream $\sv'$ s.t. $\sv'(\indx)=\sv(\indx)+\sv(\indx+1)+\sv(\indx+2)$.
On top of \lstinline{aggr},  we can easily define \lstinline{avg(n,s)}
to compute the stream of average values of \lstinline{s} in the frame of length \lstinline{n}:
\begin{lstlisting}
avg(n,s) = aggr(n,s)[/]repeat(n)  
\end{lstlisting}



%% file: wf_iff.tex
\section{\DA{Well-defined environments and equivalent streams}}\label{sect:wf}
In the semantic rules, we have left unspecified two notions: \emph{\DA{well-defined}  environments}, and \emph{\DA{equivalent streams}}. 
We provide now a formal definition in abstract terms. Then, we provide an operational definition of \DA{well-defined environments}.

Semantically, a stream $\stream$ is an infinite sequence of numeric values{, that is, a function which returns, for each index $i\geq 0$, the $i$-th element $\stream(i)$.}
Given a result $\Caps{\sv}{\mapEnv}$, we get a stream by instantiating variables in $\sv$ with streams, in a way consistent with $\mapEnv$, and evaluating operators.
To make this formal, we need some preliminary definitions.

A \emph{substitution}  {$\subst$} is a function from a finite set of variables to streams. We denote by $\eval{\sv}{\subst}$ the  stream obtained by applying $\subst$ to $\sv$, and evaluating {operators}, as formally defined below.
\begin{quote}
$\eval{\x}{\subst} =\subst(\x)$\\[1ex]
$(\eval{\cons{\nv}{\sv}}{\subst})(i) = 
\begin{cases}
\nv & i=0 \\ (\eval{\sv}{\subst})(i-1) & i\geq 1 
\end{cases}$\\[1ex]
$(\eval{\tail{\sv}}{\subst})(i) = \eval{\sv}{\subst}(i+1)\BigSpace i\geq 0$\\[1ex]
$(\eval{\PW{\sv_1}{\pwop}{\sv_2}}{\subst})(i) = \eval{\sv_1}{\subst}(i) \mathbin{\nop} \eval{\sv_2}{\subst}(i)\BigSpace i\geq 0$
\end{quote}

Given an environment $\mapEnv$ and a substitution $\subst$ with domain $\vars{\mapEnv}$, the substitution  $\AppMap{\mapEnv}{\subst}$ is defined by: 
\begin{quote}
$\AppMap{\mapEnv}{\subst}(\x) = \begin{cases}
\eval{\mapEnv(\x)}{\subst} & \x \in \dom{\mapEnv} \\
\subst(x) & {\x\in \fv{\mapEnv}}
\end{cases}$
\end{quote}
Then, a \emph{solution} of $\mapEnv$ is a substitution $\subst$ \DA{with domain $\vars{\mapEnv}$} such that $\AppMap{\mapEnv}{\subst} = \subst$. 

A closed environment $\mapEnv$ is \emph{\DA{well-defined}} if it has exactly one solution, denoted $\sol(\mapEnv)$. 
For instance, ${\{\x\mapsto\cons{1}{\x}\}}$ and ${\{\y\mapsto\cons{0}{(\y [+] \x)},\ \x\mapsto1:\x\}}$ are \DA{well-defined}{, since their unique solutions map $\x$ to the infinite stream of ones, and $\y$ to the stream of natural numbers, respectively.} Instead, for ${\{\x\mapsto1[+]\x\}}$ there are no solutions. Lastly, an environment can be undetermined{:  for instance, a substitution mapping $\x$ into an arbitrary stream} is a solution of ${\{\x\mapsto\x\}}$.

{An open environment $\mapEnv$ is \DA{well-defined} if, for each $\subst$ with domain $\fv{\mapEnv}$, it has exactly one solution $\subst'$ such that $\subst\subseteq\subst'$. For instance, the open environment  $\{\y\mapsto\cons{0}{(\y [+] \x)}\}$ is \DA{well-defined}.} 

Given a {closed} result $\Caps{\sv}{\mapEnv}$, with $\mapEnv$ \DA{well-defined}, we define its semantics by  
$\sem{\sv}{\mapEnv} = \eval{\sv}{\subst}$ for  $\subst =\sol(\mapEnv)$. Then, two stream values $\sv$ and $\sv'$ are \emph{semantically equivalent in $\mapEnv$}   if $\sem{\sv}{\mapEnv}=\sem{\sv'}{\mapEnv'}$.

\DA{We now consider the non-trivial problem of ensuring that a closed environment $\mapEnv$ is well-defined; if environments would be allowed to contain only
  the stream constructor, then it would suffice to require all non-free variables to be \emph{guarded} by the stream constructor \cite{Coquand93}.
  For instance, the environment ${\{\x\mapsto\cons{1}{\x}\}}$ satisfies such a syntactic condition, and is well-defined, while 
  in the non well-defined environment ${\{\x\mapsto\x\}}$ variable $\x$ is not guarded by the constructor.}

\DA{However, when non constructors as
  the tail and pointwise operators come into play, the fact that variables are guarded by the stream constructor no longer ensures that the environment is well-defined; let us consider for instance $\mapEnv=\{\x\mapsto \cons{0}{\tail{\x}}\}$ corresponding to the definition of \lstinline{bad_stream}
  shown in \refToSection{intro}: $\mapEnv$ is not well-defined since it admits infinite solutions (all streams starting with 0), although variable $\x$
is guarded by the stream constructor.}
  
\DA{To ensure well-defined environments a more complex check is needed:} in \refToFigure{op-wf} we provide an operational characterization of \DA{well-defined environments}. 

\begin{figure}
\begin{small}
\begin{grammatica}
\produzione{\map}{{\x_1\mapsto\nv_1 \ldots \x_n\mapsto\nv_k} \Space (n\geq0)}{map from variables to natural numbers}
\end{grammatica}
\\
\hrule 
$\begin{array}{l}
  \\
\NamedRule{main}
{\WFUpdated{\x}{\Update{\mapEnv}{\x}{\val}}{\emptyMap}}
{\wf(\mapEnv,\x,\val)}
{}
\BigSpace    
\NamedRule{wf-var}
{\WF{\mapEnv(\x)}{\Update{\map}{\x}{0}}}
{\WF{\x}{\map} }
{\x\not\in\dom\map}
\\[5ex]
\NamedRule{wf-corec}
{}
{\WF{\x}{\map} }
{\x\in\dom\map\\
\map(x)>0}
\BigSpace
{\NamedRule{wf-fv}
{}
{\WF{\x}{\map} }
{\x\not\in\dom\mapEnv}}
\\[6ex]
\NamedRule{wf-cons}
{\WF{\sv}{{\incrMap}}}
{\WF{\cons{\nv}{\sv}}{\map}}
{}
\BigSpace
\NamedRule{wf-tail}
{\WF{\sv}{{\decrMap}}}
{\WF{\tail{\sv}}{\map}}
{}
\NamedRule{wf-pw}
{\WF{\sv_1}{\map}\Space\WF{\sv_2}{\map}}
{\WF{\PW{\sv_1}{\pwop}{\sv_2}}{\map}}
{}
\end{array}$
\end{small}
\caption{Operational definition of \DA{well-defined environments}}\label{fig:op-wf}
\end{figure}

{The judgment $\wf(\mapEnv,\x,\sv)$ used in the side condition of rule \refToRule{invk} holds if $\WFUpdated{\x}{\Update{\mapEnv}{\x}{\val}}{\emptyMap}$ holds. The judgment $\WF{\sv}{\emptyMap}$ means that a result is well-defined. That is, restricting the domain of $\mapEnv$ to the variables reachable from $\sv$ (that is, either occurring in $\sv$, or, transitively, in values associated with reachable variables) we get a \DA{well-defined} environment; thus, $\wf(\mapEnv,\x,\sv)$ holds if adding the association of $\sv$ with $\x$ preserves \DA{well-definedness} of $\mapEnv$.}

The additional argument $\map$ in the judgment $\WF{\sv}{\map}$ is a map from variables to natural numbers.  We write $\incrMap$ and $\decrMap$ for the maps $\{(\x,\map(\x)+1) \mid {\x\in\dom\map}\}$, and $\{(\x,\map(\x)-1) \mid \x\in\dom\map\}$, respectively.

In rule \refToRule{main}, this map is initially empty.
In rule \refToRule{wf-var}, a variable $\x$ {defined in the environment} is added in the map, with initial value $0$, the first time it is found.  
In rule \refToRule{wf-corec}, when it is found the second time, it is checked that more constructors than tail operators have been traversed. In rule \refToRule{wf-fv}, a free variable is considered \DA{well-defined}.\footnote{Indeed, \DA{non-well-definedness} can only be detected on closed results.} In rules \refToRule{wf-cons}, \refToRule{wf-tail}, and \refToRule{wf-pw}, the value associated with a variable is incremented/decremented by one each time a constructor and tail operator are traversed, respectively.

As an example of derivation of \DA{well-definedness} and access to the $\indx$-th element, in \refToFigure{ex1} we consider the result $\Caps{\x}{\{\x\mapsto \cons{0}{(\x\ [+]\ \y)},\y\mapsto \cons{1}{\y}\}}$, obtained by evaluating the call \lstinline{nat()} with \lstinline!nat! defined as in \refToSection{examples}.  
\begin{figure}
\begin{small}
$\begin{array}{l}
\NamedRule{wf-var}
{\NamedRule{wf-cons}
{\NamedRule{{wf-pw}}
{\NamedRule{wf-corec}
{}
{\WF{\x}{\{\x\mapsto1\}}}
{}
\BigSpace
\NamedRule{wf-var}
{\NamedRule{wf-cons}
{\NamedRule{wf-corec}
{}
{\WF{\y}{\{\x\mapsto1,\y\mapsto 1\}}}
{}}
{\WF{\cons{1}{\y}}{\{\x\mapsto1,\y\mapsto0\}}}
{}}
{\WF{\y}{\{\x\mapsto1\}}}
{}}
{\WF{\x\ [+]\ \y}{\{\x\mapsto1\}}}
{}}
{\WF{\cons{0}{(\x\ [+]\ \y)}}{\{\x\mapsto0\}}}
{}}
{\WF{\x}{\emptyset}}
{}
\end{array}$
\\[4ex]
$\begin{array}{l}
\NamedRule{at-var}
{\NamedRule{at-cons-n}
{\NamedRule{at-op}
{\NamedRule{at-var}
{\begin{array}{c}\NamedRule{at-cons-0}{}{\At{\mapEnv}{\cons{0}{(\x\ [+]\ \y)}}{0}{0}}{}\\\vdots\end{array}}
{\At{\mapEnv}{\x}{\indx-1}{\indx-1}}
{}
\BigSpace
\NamedRule{at-var}
{
\begin{array}{c}\NamedRule{at-cons-0}{}{\At{\mapEnv}{\cons{1}{\y}}{0}{1}}{}\\\vdots\end{array}
}
{\At{\mapEnv}{\y}{\indx-1}{1}}
{}}
{\At{\mapEnv}{\x\ [+]\ \y}{\indx-1}{\indx}}
{}}
{\At{\mapEnv}{\cons{0}{(\x\ [+]\ \y)}}{\indx}{\indx}}
{}}
{\At{\mapEnv}{\x}{\indx}{\indx}}
{}
\end{array}$
\end{small}
\caption{Derivations for $\mapEnv=\Caps{\x}{\{\x\mapsto \cons{0}{(\x\ [+]\ \y)},\ y\mapsto \cons{1}{\y}\}}$.}\label{fig:ex1}
\end{figure}

In \refToFigure{ex2} in the Appendix we consider 
a trickier example, that is, the result $\Caps{\x}{\{\x\mapsto \cons{0}{\cons{1}{\tail{(\cons{2}{\tail{\x}})}}}\}}$. Its semantics is the stream $0,1,1,1,\ldots$.  %

We show now that \DA{well-definedness} of a result is a necessary and sufficient condition for termination of access to an arbitrary index.
To formally express and prove this statement, we introduce some definitions and notations.

First of all, since  the numeric value obtained as result is not relevant for the following technical treatment, for simplicity we will write $\AT{\sv}{\indx}$ rather than $\At{\mapEnv}{\sv}{\indx}{\nv}$. We call \emph{derivation} an either finite or infinite proof tree.

We write $\WF{\sv'}{\map'}\premise\WF{\sv}{\map}$ to mean that $\WF{\sv'}{\map'}$ is a premise of a {(meta-)rule} where $\WF{\sv}{\map}$ is the consequence, and $\premisestar$ for the reflexive and transitive closure of this relation. Moreover, 
$\WF{\x}{\map'}\premisestarfirst{\varset}\WF{\sv}{\map}$, with $\x\not\in\varset$, means that in the path there can be nodes of shape $\WF{\y}{\_}$ only for $\y\in\varset$ and non-repeated. We use analogous notations for the judgment $\AT{\sv}{\indx}$. 

\begin{lemma}\label{lemma:basics}\
\begin{enumerate}
\item\label{iii} If $\AT{\x}{\indx'}\premisestarfirst{\varset}\AT{\sv}{\indx}$, then $\AT{\x}{\indx'+k}\premisestarfirst{\varset}\AT{\sv}{\indx+k}$, for each $k\geq 0$.
\item\label{i} A judgment $\WF{\sv}{\emptyset}$ has no derivation
iff the following condition holds:\\
\begin{tabular}{ll}
\refToRule{wf-stuck}&$\WF{\x}{\map'}\premisestarfirst{\varset'} \WF{\mapEnv(\x)}{\Update{\map}{\x}{0}}\premise\WF{\x}{\map}\premisestarfirst{\varset}\WF{\sv}{\emptyMap}$\\
&for some $\x\in\dom{\mapEnv}$, $\varset',\varset$, and $\map',\map$ s.t.\ $\x\not\in\dom{\map}, \map'(\x)=k\leq 0$.
\end{tabular}
\item\label{ii} The derivation of $\AT{\sv}{\jndx}$ 
is infinite iff the following condition holds:\\
\begin{tabular}{ll}
\refToRule{at-$\infty$}&$\AT{\x}{\indx+k}\premisestarfirst{\varset'}\AT{\mapEnv(\x)}{\indx}\premise\AT{\x}{\indx}\premisestarfirst{\varset}\AT{\sv}{\jndx}$\\
&for some $\x\in\dom{\mapEnv}$, $\varset',\varset$, and $\indx, k \geq 0$.
\end{tabular}
\end{enumerate}
\end{lemma} 
\begin{proof}\
\begin{enumerate}
\item Immediate by induction on the rules.
\item For each $\WF{\sv}{\map}$ there is exactly one applicable rule, unless in the case $\WF{\x}{\map}$ with $\map(\x)\leq 0$. Since $\vars{\sv}$ is a finite set, the derivation cannot be infinite. 
Hence, there is no derivation for $\WF{\sv}{\emptyset}$  iff there is a finite path from $\WF{\sv}{\emptyset}$ of judgments on variables in $\dom{\mapEnv}$, and a (first) repeated variable, that is, of the shape below, where $\x\not\in\dom{\map}$ and $\map'(\x)\leq 0$.
\begin{small}
$$\WF{\x}{\map'}\WF{\x_n}{\_}\ldots\WF{\x_1}{\_}\WF{\x}{\map}\WF{\y_m}{\_}\ldots\WF{\y_1}{\_}\ldots\WF{\sv}{\emptyset}$$
\end{small}
That is, condition \refToRule{wf-stuck} holds, with $\varset'=\{\x_1,\ldots,\x_n\}$, and $\varset=\{\y_1,\ldots,\y_m\}$.
\item For each $\AT{\sv}{\indx}$ there is exactly one applicable rule, unless in the case $\AT{\x}{\indx}$ with $\x\not\in\dom{\mapEnv}$. Moreover, since $\mapEnv$ has finite domain, the derivation $\AT{\sv}{\jndx}$ is infinite iff there is an infinite path from $\AT{\sv}{\jndx}$ of judgments on variables in $\dom{\mapEnv}$, and a (first) repeated variable with a greater or equal index, hence, thanks to \refToLemma{basics}-(\ref{iii}), of the shape below, where $m,n,k\geq 0$:
\begin{small}
$$\ldots\AT{\x_n}{\indx_n+k}\ldots\AT{\x_1}{\indx_1+k}\AT{\x}{\indx+k}\AT{\x_n}{\indx_n}\ldots$$
$$\ldots\AT{\x_1}{\indx_1}\AT{\x}{\indx}\AT{\y_m}{\jndx_m}\ldots\AT{\y_1}{\jndx_1}\ldots\AT{\sv}{\jndx}$$
\end{small}
That is, condition \refToRule{at-$\infty$} holds, with $\varset'=\{\x_1,\ldots,\x_n\}$, and $\varset=\{\y_1,\ldots,\y_m\}$.
\end{enumerate}
\end{proof}

\begin{lemma}\label{lemma:second-occurrence} For $\x\in\dom{\map}$, the following conditions are equivalent:
\begin{enumerate}
\item $\AT{\x}{\indx'}\premisestarfirst{\varset}\AT{\sv}{\indx}$ for some $\indx',\indx$
\item $\WF{\x}{\map'}\premisestarfirst{\varset}\WF{\sv}{\map}$ for some $\map'$ such that $\map'(\x) = \map(\x)+\indx-\indx'$.
\end{enumerate}
\end{lemma}
\begin{proof}\
\begin{description}
\item[1$\Rightarrow 2$]
The proof is by induction on the length of the path in $\AT{\x}{\indx'}\premisestarfirst{\varset}\AT{\sv}{\indx}$. 
\begin{description}
\item[Base] The length of the path is $0$, hence we have $\AT{\x}{\indx}\premisestarfirst{\emptyset}\AT{\x}{\indx}$. We also have  $\WF{\x}{\map}\premisestarfirst{\emptyset}\WF{\x}{\map}$, and $\map(\x)=\map(\x)+\indx-\indx$, as requested.
\item[Inductive step] By cases on the rule applied to derive $\AT{\sv}{\indx}$.
\begin{description}
\item[\refToRule{at-var}] We have $\AT{\y}{\indx}$, with $\y\neq\x$ since the length of the path is $>0$, and $\AT{\x}{\indx'}\premisestarfirst{\chi\setminus\{\y\}}\AT{\mapEnv(\y)}{\indx}$. 
Moreover, we can derive $\WF{\y}{\map}$ by rule \refToRule{wf-var}, and by inductive hypothesis we also have  $\WF{\x}{\map'}\premisestarfirst{\chi\setminus\{\y\}}\WF{\mapEnv(\y)}{\Update{\map}{\y}{0}}$, and $\map'(\x)=\Update{\map}{\y}{0}(\x)+\indx-\indx'$, hence we get the thesis.
\item[\refToRule{at-cons-0}] Empty case, since the derivation for $\AT{\cons{\nv}{\sv}}{0}$ does not contain a node $\AT{\x}{\indx'}$.
\item[\refToRule{at-cons}] We have $\AT{\cons{\nv}{\sv}}{\indx}$, and $\AT{\x}{\indx'}\premisestarfirst{\varset}\AT{\sv}{\indx-1}$. 
Moreover, we can derive $\WF{\cons{\nv}{\sv}}{\map}$ by rule \refToRule{wf-cons}, and by inductive hypothesis we also have  $\WF{\x}{\map'}\premisestarfirst{\varset}\WF{\sv}{\incrMap}$, with ${\map'(\x)=\incrMap(\x)+(\indx-1)-\indx'}$, hence we get the thesis.
\item[\refToRule{at-tail}] This case is symmetric to the previous one.
\item[\refToRule{at-pw}] We have $\AT{\PW{\sv_1}{\pwop}{\sv_2}}{\indx}$, and either $\AT{\x}{\indx'}\premisestarfirst{\varset}\AT{\sv_1}{\indx}$, or \linebreak ${\AT{\x}{\indx'}\premisestarfirst{\varset}\AT{\sv_2}{\indx}}$. Assume the first case holds, the other is analogous.
Moreover, we can derive $\WF{\PW{\sv_1}{\pwop}{\sv_2}}{\map}$ by rule \refToRule{wf-pw}, and by inductive hypothesis we also have  $\WF{\x}{\map'}\premisestarfirst{\varset}\WF{\sv_1}{\map}$, with $\map'(\x)=\map(\x)+\indx-\indx'$, hence we get the thesis.
\end{description}
\end{description}
\item[2$\Rightarrow$ 1] The proof is by induction on the length of the path in ${\WF{\x}{\map'}\premisestar\WF{\sv}{\map}}$.
\begin{description}
\item[Base] The length of the path is $0$, hence we have $\WF{\x}{\map}\premisestarfirst{\emptyset}\WF{\x}{\map}$. We also have, for an arbitrary $\indx$, $\AT{\x}{\indx}\premisestarfirst{\emptyset}\AT{\x}{\indx}$, and $\map(\x)=\map(\x)+\indx-\indx$, as requested. 
\item[Inductive step] By cases on the rule applied to derive $\WF{\sv}{\map}$.
\begin{description}
\item[\refToRule{wf-var}] We have $\WF{\y}{\map}$, with $\y\notin\dom{\map}$, $\y\neq\x$ since ${\x\in\dom{\map}}$, and $\WF{\x}{\map'}\premisestarfirst{\varset\setminus\{\y\}}\WF{\mapEnv(\y)}{\Update{\map}{\y}{0}}$.
By inductive hypothesis we have $\AT{\x}{\indx'}\premisestarfirst{\varset\setminus\{\y\}}\AT{\mapEnv(\y)}{\indx}$ for some $\indx',\indx$ such that $\map'(\x) = \map(\x)+\indx-\indx'$. 
Moreover, since $\y\in\dom{\mapEnv}$, $\AT{\mapEnv(\y)}{\indx}\premise\AT{\y}{\indx}$ by rule \refToRule{at-var}, hence we get $\AT{\x}{\indx'}\premisestarfirst{\chi}\AT{\y}{\indx}$.

\item[\refToRule{wf-corec}] Empty case, since the derivation for $\WF{\y}{\map}$ would not contain a node $\WF{\x}{\map}$.
\item[\refToRule{wf-fv}] Empty case, since the derivation for $\WF{\y}{\map}$ would not contain a node $\WF{\x}{\map}$.
\item[\refToRule{wf-cons}] We have $\WF{\cons{\nv}{\sv}}{\map}$, and $\WF{\x}{\map'}\premisestarfirst{\varset}\WF{\sv}{\incrMap}$. By inductive hypothesis we have $\AT{\x}{\indx'}\premisestarfirst{\varset}\AT{\sv}{\indx}$ for some $\indx',\indx$ such that $\map'(\x) = \incrMap(\x)+\indx-\indx'$. Moreover, $\AT{\sv}{\indx}\premise\AT{\cons{\nv}{\sv}}{\indx+1}$ by rule \refToRule{at-cons-n}, hence we get $\AT{\x}{\indx'}\premisestarfirst{\chi}\AT{\cons{\nv}{\sv}}{\indx+1}$ with $\map'(\x) = \map(\x)+\indx+1-\indx'$, as requested.
\item[\refToRule{wf-tail}] We have $\WF{\tail{\sv}}{\map}$, and $\WF{\x}{\map'}\premisestarfirst{\varset}\WF{\sv}{\decrMap}$. By inductive hypothesis we have $\AT{\x}{\indx'}\premisestarfirst{\varset}\AT{\sv}{\indx}$ for some $\indx',\indx$ such that $\map'(\x) = \decrMap(\x)+\indx-\indx'$. We can assume $\indx>0$ thanks to \refToLemma{basics}-(\ref{iii}). Hence, $\AT{\sv}{\indx}\premise\AT{\cons{\nv}{\sv}}{\indx-1}$ by rule \refToRule{at-tail}, hence we get $\AT{\x}{\indx'}\premisestarfirst{\chi}\AT{\tail{\sv}}{\indx-1}$ with $\map'(\x) = \map(\x)+\indx-1-\indx'$, as requested.
\item[\refToRule{wf-pw}] We have $\WF{\PW{\sv_1}{\pwop}{\sv_2}}{\map}$, and either $\WF{\x}{\map'}\premisestarfirst{\varset}\WF{\sv_1}{\map}$, or ${\WF{\x}{\map'}\premisestarfirst{\varset}\WF{\sv_2}{\map}}$. Assume the first case holds, the other is analogous.
By inductive hypothesis we  have  $\AT{\x}{\indx'}\premisestarfirst{\varset}\AT{\sv_1}{\indx}$, for some $\indx',\indx$ such that  $\map'(\x)=\map(\x)+\indx-\indx'$. Moreover, we can derive $\AT{\PW{\sv_1}{\pwop}{\sv_2}}{\indx}$ by rule \refToRule{at-pw},  hence we get the thesis.
\end{description}
\end{description}
\end{description}
\end{proof}

\begin{lemma}\label{lemma:first-occurrence} For $\x\not\in\dom{\map}$, the following conditions are equivalent:
\begin{enumerate}
\item $\AT{\x}{\indx'}\premisestarfirst{\varset}\AT{\sv}{\indx}$ for some $\indx',\indx$
\item  $\WF{\x}{\map'}\premisestarfirst{\varset}\WF{\sv}{\map}$ for some $\map'$ such that $\x\not\in\dom{\map'}$.
\end{enumerate}
\end{lemma}
\begin{proof}
Easy variant of the proof of \refToLemma{second-occurrence}.
\end{proof}

\begin{theorem}\label{theo:iff}
$\WF{\sv}{\emptyMap}$ is derivable iff, for all $\jndx$, $\AT{\sv}{\jndx}$ either has no derivation or a finite derivation.
\end{theorem}
\begin{proof} We prove that $\AT{\sv}{\jndx}$ has an infinite derivation for some $\jndx$ iff $\WF{\sv}{\emptyMap}$ has no derivation.
\begin{description}
\item[$\Rightarrow$] By \refToLemma{basics}-(\ref{ii}),we have that the following condition holds:
\begin{quote}
\begin{tabular}{ll}
\refToRule{at-$\infty$}&$\AT{\x}{\indx+k}\premisestarfirst{\varset'}\AT{\mapEnv(\x)}{\indx}\premise\AT{\x}{\indx}\premisestarfirst{\varset}\AT{\sv}{\jndx}$\\
&for some $\x\in\dom{\mapEnv}$, $\varset',\varset$, and $\indx, k \geq 0$.
\end{tabular}
\end{quote}
Then, starting from the right, by \refToLemma{first-occurrence} we have ${\WF{\x}{\map}\premisestarfirst{\varset}\WF{\sv}{\emptyMap}}$ for some $\map$ such that $\x\not\in\dom{\map}$; by rule \refToRule{wf-var} we have \linebreak ${\WF{\mapEnv(\x)}{\Update{\map}{\x}{0}}\premise\WF{\x}{\map}}$, and finally by \refToLemma{second-occurrence} we have:\\ 
\begin{tabular}{ll}
\refToRule{wf-stuck}&$\WF{\x}{\map'}\premisestarfirst{\varset'} \WF{\mapEnv(\x)}{\Update{\map}{\x}{0}}\premise\WF{\x}{\map}\premisestarfirst{\varset}\WF{\sv}{\emptyMap}$\\
&for some $\x\in\dom{\mapEnv}$, $\varset',\varset$, and $\map',\map$ s.t.\ $\x\not\in\dom{\map}, \map'(\x)=k\leq 0$.
\end{tabular}
\\ 
hence we get the thesis.
\item[$\Leftarrow$] By \refToLemma{basics}-(\ref{i}),we have that the condition \refToRule{wf-stuck} above holds.
Then, starting from the left, by \refToLemma{second-occurrence} we have $\AT{\x}{\indx'}\premisestarfirst{\varset'}\AT{\mapEnv(\x)}{\indx}$ for some $\indx',\indx$ such that $\indx-\indx'=k\leq 0$; 
by rule \refToRule{at-var} we have $\AT{\mapEnv(\x)}{\indx}\premise\AT{\x}{\indx}$, and by \refToLemma{first-occurrence} we have $\AT{\x}{\jndx'}\premisestarfirst{\varset}\AT{\sv}{\jndx}$ for some $\jndx',\jndx$.
If $\indx=\jndx'+h$, $h \geq 0$, then by \refToLemma{basics}-(\ref{iii}) we have\\
\centerline{
$\AT{\x}{\indx+k}\premisestarfirst{\varset'}\AT{\mapEnv(\x)}{\indx}\premise\AT{\x}{\indx}\premisestarfirst{\varset}\AT{\sv}{\jndx+h}$}
If $\jndx'=\indx+h$, $h \geq 0$, then by \refToLemma{basics}-(\ref{iii}) we have\\
\centerline{$\AT{\x}{\indx+k+h}\premisestarfirst{\varset'}\AT{\mapEnv(\x)}{\indx}\premise\AT{\x}{\indx+h}\premisestarfirst{\varset}\AT{\sv}{\jndx}$.}
In both cases, the derivation of $\AT{\sv}{\jndx}$ is infinite. 
\end{description}
\end{proof}

%% file: conclu.tex
\section{Related and future work}\label{sect:conclu}
As  mentioned in \refToSection{intro}, our approach extends regular corecursion, where the semantics keeps track of method/function calls. Regular corecursion originated from \emph{co-SLD resolution} \cite{Simon06,SimonBMG07,Ancona13,AnconaDovier15}, where already considered goals (up to unification), called \emph{coinductive hypotheses}, are considered successful. Language constructs that support this programming style have also been proposed in the functional \cite{Jeannin17} and object-oriented \cite{AnconaZ12,AnconaBDZ20} paradigm. 

There have been a few attempts of extending the expressive power of regular corecursion. Notably, \emph{structural resolution} \cite{KomendantskayaJS16,KomendantskayaPS16} is a proposed operational semantics for logic programming where infinite derivations that cannot be built in finite time are generated lazily, and only partial answers are shown. Another approach is the work on infinite trees \cite{Courcelle83}, where Courcelle introduces algebraic trees and equations as generalizations of regular ones.

\DA{For what concerns the operators considered in the calculus and some examples, our main sources of inspiration have been the works of
Rutten \cite{Rutten05}, where a coinductive calculus of streams of real numbers is defined,
and Hinze \cite{Hinze10}, where a calculus of generic streams is defined in a constructive way and implemented in Haskell.}

\DA{The problem of ensuring well-defined corecursive definitions has been also considered in the context of
  type theory and proof assistants. We have shown in \refToSection{wf} that simple guarded definitions \cite{Coquand93} do not work properly
  in case values are allowed to contain non constuctors as the tail operator; a
  more complex approach based on a type system has been proposed by Sacchini \cite{Sacchini13} for an extension of the calculus of constructions
  which is more expressive than what considered here; however, as opposed to what happens with the judgment $\wf$ defined in \refToSection{wf},
  corecursive  calls  to  the  result  of  an  application of tail are never well-typed even in case of well-defined streams as happens
  for the definition of \lstinline{fib} as given in \refToSection{examples}.}
  
\PB{Lastly, D'Angelo et al. presented LOLA \cite{LOLA05}, a specification language for runtime monitoring that manipulates streams. The general idea behind the framework is to generate a set of output streams,
starting from a given set of input streams. 
The main difference with respect to our work is that LOLA only allows streams with a finite number of elements.
In this framework, well-formedness is checked by relying on a dependency graph, which keeps track of relations between the processed streams. The vertices of this graph are the streams, while the edges represent the dependencies between them. Each edge is weighted with a value $\omega$ to point the fact that a stream depends on another one shifted by $\omega$ positions. Then, the well-formedness constraint is that each closed-walk inside the graph must have a total weight different from $0$. These syntactic constraints appear to be very similar to the approach we used for predicate $\wf$.}

Our main technical result is \refToTheorem{iff}, stating that passing the \DA{well-definedness} check performed at runtime for each function call is necessary and sufficient to prevent non-termination in accessing elements in the resulting stream at an arbitrary index. In future work, we plan to also prove soundness of the operational \DA{well-definedness} with respect to its abstract definition. Completeness does not hold, as shown by the example
\lstinline!zeros() = repeat(0) [*] zeros()! 
which is not \DA{well-defined} operationally, but admits as unique solution the stream of all zeros.
On the other hand, the simplest operational characterization of equivalence of stream values is syntactic equivalence. This works for all the examples presented in this paper, but is, again, not complete with respect to the abstract definition, as illustrated below:

\begin{lstlisting}
first(s) = s(0):first(s) // works with syntactic equivalence
first2(s) = s(0):first2(s(0):s^) // does not work
\end{lstlisting}

Indeed, we get an infinite derivation for \lstinline!first2(repeat(1))!
  with call traces of increasing shape
  ${\{\mbox{\lstinline!first2($x$)!}\mapsto \y_1,\mbox{\lstinline!first2(1:$x$^)!}\mapsto \y_2,\mbox{\lstinline!first2(1:(1:$x$^)^)!}\mapsto \y_3,\ldots\}}$ in the environment $\{\x\mapsto\cons{1}{\x}\}$. In future work we plan to investigate more expressive operational characterizations of equivalence.

%% file: appendix.tex
\section{Examples of derivations}

\begin{figure}
\begin{small}
$\begin{array}{l}
\f()=\g()\\
g()=1:\f()
\\[6ex]
\NamedRule{invk}
{\NamedRule{invk}
{\NamedRule{cons}
{\NamedRule{val}
{}
{\opsem{1}{\emptyset}{\{\f()\mapsto\x,\g()\mapsto\y\}}{1}{\emptyset}}
{}
\BigSpace
\NamedRule{corec}
{}
{\opsem{\f()}{\emptyset}{\{\f()\mapsto\x,\g()\mapsto\y\}}{\x}{\emptyset}}
{}}
{\opsem{1:\f()}{\emptyset}{\{\f()\mapsto\x,\g()\mapsto\y\}}{1:\x}{\emptyset}}
{}}
{\opsem{\g()}{\emptyset}{\{\f()\mapsto\x\}}{\y}{\{\y\mapsto1:\x\}}}
{}}
{\opsem{\f()}{\emptyset}{\emptyset}{\x}{\{\x\mapsto\y,\y\mapsto1:\x\}}}
{}
\end{array}$
\end{small}
\caption{Example of derivation}\label{fig:derivation1}
\end{figure}

\begin{figure}
\begin{small}
$\begin{array}{l}
\f()=\g(2:\f())\\
g(s)=1:s
\\[6ex]
\NamedRule{invk}
{\NamedRule{args}
{\NamedRule{cons}
{\NamedRule{val}
{}
{\opsem{2}{\emptyset}{\{\f()\mapsto\x\}}{2}{\emptyset}}
{}
\BigSpace
\NamedRule{corec}
{}
{\opsem{\f()}{\emptyset}{\{\f()\mapsto\x\}}{\x}{\emptyset}}
{}}
{\opsem{2:\f()}{\emptyset}{\{\f()\mapsto\x\}}{2:\x}{\emptyset}}
{}}
{\opsem{\g(2:\f())}{\emptyset}{\{\f()\mapsto\x\}}{\y}{\{\y\mapsto1:2:\x\}}}
{}
\BigSpace
T_1}
{\opsem{\f()}{\emptyset}{\emptyset}{\x}{\{\x\mapsto\y,\y\mapsto1:2:\x\}}}
{}
\\[8ex]
T_1=
\NamedRule{invk}
{\NamedRule{val}
{}
{\opsem{1:2:\x}{\emptyset}{\{\g(2:\x)\mapsto\y,\f()\mapsto\x\}}{1:2:\x}{\emptyset}}
{}}
{\opsem{\g(2:\x)}{\emptyset}{\{\f()\mapsto\x\}}{\y}{\{\y\mapsto1:2:\x\}}}
{}
\end{array}$
\end{small}
\caption{Example of derivation}\label{fig:derivation2}
\end{figure}

\begin{figure}
\begin{small}
$\begin{array}{l}
\undef()=(\undef()(0)):\undef()
\\[6ex]
\NamedRule{invk}
{\NamedRule{cons}
{T_1
\BigSpace
\NamedRule{corec}
{}
{\opsem{\undef()}{\emptyset}{\{\undef()\mapsto\x\}}{\x}{\emptyset}}
{}}
{\opsem{(\undef()(0)):\undef()}{\emptyset}{\{\undef()\mapsto\x\}}{?}{?}}
{}
}
{\opsem{\undef()}{\emptyset}{\emptyset}{?}{?}}
{}
\\[8ex]
T_1=
\NamedRule{at}
{\NamedRule{corec}
{}
{\opsem{\undef()}{\emptyset}{\{\undef()\mapsto\x\}}{\x}{\emptyset}}
{}
\BigSpace
\NamedRule{val}
{}
{\opsem{0}{\emptyset}{\{\undef()\mapsto\x\}}{0}{\emptyset}}
{}}
{\opsem{\undef()(0)}{\emptyset}{\{\undef()\mapsto\x\}}{?}{?}}
{\At{\emptyset}{\x}{0}{?}}
\end{array}$
\end{small}
\caption{Example of stuck derivation}\label{fig:stuck_derivation}
\end{figure}

\begin{figure}
\begin{small}
$\begin{array}{l}
\NamedRule{wf-var}
{\NamedRule{wf-cons}
{\NamedRule{wf-cons}
{\NamedRule{wf-tail}
{\NamedRule{wf-cons}
{\NamedRule{wf-tail}
{\NamedRule{wf-corec}
{}
{\WF{\x}{\{x\mapsto 1\}}}
{}}
{\WF{\tail{\x}}{\{x \mapsto 2\}}}
{}}
{\WF{\cons{2}{\tail{\x}}}{\{x\mapsto1\}}}
{}}
{\WF{\tail{(\cons{2}{\tail{\x}})}}{\{x\mapsto 2\}}}
{}}
{\WF{\cons{1}{\tail{(\cons{2}{\tail{\x}})}}}{\{x\mapsto 1\}}}
{}}
{\WF{\cons{0}{\cons{1}{\tail{(\cons{2}{\tail{\x}})}}}}{\{x\mapsto0\}}}
{}}
{\WF{\x}{\emptyset}}
{}
\end{array}$
\\[8ex]
$\begin{array}{l}
\NamedRule{at-var}
{\NamedRule{at-cons-n}
{\NamedRule{at-cons-n}
{\NamedRule{at-tail}
{\NamedRule{at-cons-n}
{\NamedRule{at-tail}
{\NamedRule{at-var}
{\begin{array}{c}\NamedRule{at-cons-0}{}{\At{\mapEnv}{\cons{1}{\tail{(\cons{2}{\tail{\x}})}}}{0}{1}}{}\\\vdots\end{array}}
{\At{\mapEnv}{\x}{\indx-1}{{1}}}
{}}
{\At{\mapEnv}{\tail{\x}}{\indx-2}{{1}}}
{}}
{\At{\mapEnv}{2:\tail{\x}}{\indx-1}{{1}}}
{}}
{\At{\mapEnv}{\tail{(2:\tail{\x})}}{\indx-2}{{1}}}
{}}
{\At{\mapEnv}{1:\tail{(2:\tail{\x})}}{\indx-1}{{1}}}
{}}
{\At{\mapEnv}{0:1:\tail{(2:\tail{\x})}}{\indx}{{1}}}
{}}
{\At{\mapEnv}{\x}{\indx}{{1}}}
{}
\end{array}$
\end{small}
\caption{Reduction of $\WF{\x}{\emptyset}$ and $\At{\mapEnv}{\x}{\indx}{{1}}$ with $\mapEnv=\{\x\mapsto 0:1:\tail{(2:\tail{\x})}\}$ and {$\indx>1$}.}\label{fig:ex2}
\end{figure}

%% file: main.bbl
\begin{thebibliography}{10}

\bibitem{Ancona13}
Davide Ancona.
\newblock Regular corecursion in {P}rolog.
\newblock {\em Computer Languages, Systems {\&} Structures}, 39(4):142--162,
  2013.

\bibitem{AnconaBDZ20}
Davide Ancona, Pietro Barbieri, Francesco Dagnino, and Elena Zucca.
\newblock Sound regular corecursion in {coFJ}.
\newblock In Robert Hirschfeld and Tobias Pape, editors, {\em ECOOP'20 -
  Object-Oriented Programming}, volume 166 of {\em LIPIcs}, pages 1:1--1:28.
  Schloss Dagstuhl - Leibniz-Zentrum f{\"u}r Informatik, 2020.

\bibitem{AnconaDovier15}
Davide Ancona and Agostino Dovier.
\newblock A theoretical perspective of coinductive logic programming.
\newblock {\em Fundamenta Informaticae}, 140(3-4):221--246, 2015.

\bibitem{AnconaZ12}
Davide Ancona and Elena Zucca.
\newblock Corecursive {F}eatherweight {J}ava.
\newblock In {\em FTfJP'12 - Formal Techniques for Java-like Programs}, pages
  3--10. ACM Press, 2012.

\bibitem{Coquand93}
Thierry Coquand.
\newblock Infinite objects in type theory.
\newblock In {\em Types for Proofs and Programs, International Workshop
  TYPES'93, Nijmegen, The Netherlands, May 24-28, 1993, Selected Papers}, pages
  62--78, 1993.

\bibitem{Courcelle83}
Bruno Courcelle.
\newblock Fundamental properties of infinite trees.
\newblock {\em Theoretical Computer Science}, 25:95--169, 1983.

\bibitem{DagninoAZ20}
Francesco Dagnino, Davide Ancona, and Elena Zucca.
\newblock Flexible coinductive logic programming.
\newblock {\em Theory and Practice of Logic Programming}, 20(6):818--833, 2020.
\newblock Issue for ICLP 2020.

\bibitem{LOLA05}
Ben D'Angelo, Sriram Sankaranarayanan, C{\'{e}}sar S{\'{a}}nchez, Will
  Robinson, Bernd Finkbeiner, Henny~B. Sipma, Sandeep Mehrotra, and Zohar
  Manna.
\newblock {LOLA:} runtime monitoring of synchronous systems.
\newblock In {\em 12th International Symposium on Temporal Representation and
  Reasoning {(TIME} 2005)}, pages 166--174, 2005.

\bibitem{Hinze10}
Ralf Hinze.
\newblock Concrete stream calculus: An extended study.
\newblock {\em Journal of Functional Programming}, 20(5–6):463–535, 2010.

\bibitem{JeanninK12}
Jean{-}Baptiste Jeannin and Dexter Kozen.
\newblock Computing with capsules.
\newblock {\em Journal of Automata, Languages and Combinatorics},
  17(2-4):185--204, 2012.

\bibitem{Jeannin17}
Jean-Baptiste Jeannin, Dexter Kozen, and Alexandra Silva.
\newblock {CoCaml}: Functional programming with regular coinductive types.
\newblock {\em Fundamenta Informaticae}, 150:347--377, 2017.

\bibitem{KomendantskayaJS16}
Ekaterina Komendantskaya, Patricia Johann, and Martin Schmidt.
\newblock A productivity checker for logic programming.
\newblock In Manuel~V. Hermenegildo and Pedro L{\'{o}}pez{-}Garc{\'{\i}}a,
  editors, {\em Logic-Based Program Synthesis and Transformation - {LOPSTR}
  2016, Revised Selected Papers}, volume 10184 of {\em Lecture Notes in
  Computer Science}, pages 168--186. Springer, 2016.

\bibitem{KomendantskayaPS16}
Ekaterina Komendantskaya, John Power, and Martin Schmidt.
\newblock Coalgebraic logic programming: from semantics to implementation.
\newblock {\em J. Log. Comput.}, 26(2):745--783, 2016.

\bibitem{Rutten05}
Jan J. M.~M. Rutten.
\newblock A coinductive calculus of streams.
\newblock {\em Mathematical Structures in Computer Science}, 15(1):93--147,
  2005.

\bibitem{Sacchini13}
Jorge~Luis Sacchini.
\newblock Type-based productivity of stream definitions in the calculus of
  constructions.
\newblock In {\em 28th Annual {ACM/IEEE} Symposium on Logic in Computer
  Science, {LICS} 2013, New Orleans, LA, USA, June 25-28, 2013}, pages
  233--242, 2013.

\bibitem{Simon06}
Luke Simon.
\newblock {\em Extending logic programming with coinduction}.
\newblock PhD thesis, University of Texas at Dallas, 2006.

\bibitem{SimonBMG07}
Luke Simon, Ajay Bansal, Ajay Mallya, and Gopal Gupta.
\newblock Co-logic programming: Extending logic programming with coinduction.
\newblock In Lars Arge, Christian Cachin, Tomasz Jurdzinski, and Andrzej
  Tarlecki, editors, {\em Automata, Languages and Programming, 34th
  International Colloquium, {ICALP} 2007}, volume 4596 of {\em Lecture Notes in
  Computer Science}, pages 472--483. Springer, 2007.

\end{thebibliography}
